\documentclass[12pt]{amsart}  
\usepackage{amscd} 
\usepackage{mathrsfs}
\newtheorem{theorem}{Theorem}[section]
\newtheorem{lemma}[theorem]{Lemma}
 
\theoremstyle{definition}

\newtheorem{remark}[theorem]{Remark} 
\numberwithin{equation}{section}

\def\RE{\mathbb R}
\def\CO{\mathbb C}
\def\H{\mathscr H}
\def\D{\mathscr D}
\def\K{\mathscr K}
\def\h{\mathfrak h}
\def\Z{\mathbb Z}

\def\R{\mathscr R}

\def\p{\par\noindent}
\def\ttau{\tilde\tau}
\def\uno{\mathsf 1}
\def\n{|\!|\!|}

\def\op{\underset{k\in\Z}{\oplus}}
\def\opn{\underset{n\in\NA}{\oplus}}

\def\be{\begin{equation}}
\def\ee{\end{equation}}

\def\ci{\mathbb T}
\def\M{\mathbb M}
\def\NA{\mathbb N}
\begin{document}

\title[Direct sums of trace maps and self-adjoint extensions]
{Direct sums of trace maps and self-adjoint extensions 
}

\author{Andrea Posilicano}
\address{DiSAT - Sezione di Matematica,  Universit\`a
dell'Insubria, I-22100 Como, Italy}

\email{posilicano@uninsubria.it}

\begin{abstract} We give a simple criterion so that a countable infinite direct sum of trace (evaluation)  maps is a trace map. An application to the theory of self-adjoint extensions of direct sums of symmetric operators is provided; this gives an alternative approach to results recently obtained by Malamud-Neidhardt and Kostenko-Malamud using regularized direct sums of boundary triplets.  
\end{abstract}

{\maketitle }

\begin{section}{Introduction}
We begin with a simple example. Let $\Delta_{0}=\frac{\partial^{2}}{\partial x^{2}}+\frac{\partial^{2}}{\partial\theta^{2}}$ be the Laplace-Beltrami operator on the  two-dimensional cylinder $\M_0 :=\RE_{+}\times\ci$ with respect to the flat Riemannian metric $g_{0}
=\left(\begin{smallmatrix}1&0\\0&1\end{smallmatrix}\right)$. Its minimal 
realization with domain $C^{\infty}_{c}(\M_0 )$ is symmetric and negative as a linear operator in the Hilbert space $L^{2}(\M_0 )= L^{2}(\RE_{+})\otimes L^{2}(\ci)$. We denote its Friedrichs' self-adjoint extension by $\Delta_{0}^{D}$; it corresponds to imposing Dirichlet boundary conditions at the boundary $\ci$, i.e. 
$\D(\Delta_{0}^{D})=\{u\in H^{2}(\M_0 ): \lim_{x\downarrow 0}\,u(x,\theta)=0\}$. Here $H^{2}(\M_0 )$ is the usual Sobolev-Hilbert space of order two. Let us denote by $H^{s}(\ci)$ the (fractional) Sobolev-Hilbert space 
of square-integrable functions $f$ on the $1$-dimensional torus $\ci$ such that $\sum_{k\in\Z}|k|^{2s}|\hat f_{k}|^{2}<+\infty$, where $\hat f_{k}$ is the usual Fourier coefficient $\hat f_{k}:=\frac1{\sqrt{2\pi}}\int_{\ci}e^{-ik\theta}f(\theta)\,d\theta$.  
Then $\gamma_0 :\D(\Delta_{0}^D)\to H^{\frac12}(\ci)$, the unique continuous linear map which on regular functions acts by 
$$
\gamma_0  u(\theta)=\lim_{x\downarrow 0}\,\frac{\partial u}{\partial x}(x,\theta)\,,
$$ is a concrete example of what we call an {\it abstract trace map} (see the next section),  i.e. $\gamma_0 $ is continuous (w.r.t. graph norm), surjective and its kernel is dense in $L^{2}(\M_0 )$. By partial Fourier transform with respect to the angular variable one gets 
$$L^{2}(\M_0 )=\op L^{2}(\RE_{+})\,,\quad\Delta_{0}^{D}=\op d^{2}_k\,,$$ 
where 
$$
d^{2}_{k}:\D(d^{2}_{k})\subset L^{2}(\RE_{+})\to L^{2}(\RE_{+})\,,
\quad
d^{2}_{k}f:=f''-k^{2}f\,,
$$
$$
\D(d^{2}_{k})=\D_{0}:=
\{f\in L^{2}(\RE_{+})\cap C^{1}(\overline\RE_{+}): f''\in L^{2}(\RE_{+})\,,\ f(0)=0\}\,.
$$
On $\D_{0}$ one can define the trace map $$
\hat\gamma_0 :\D_{0}\to\CO\,,\quad \hat\gamma_0  f:=f'(0)\,,
$$
which is bounded,  surjective and with a kernel dense in $L^{2}(\RE_{+})$. Moreover $\hat\gamma_{0}$ is bounded uniformly in $k\in\Z$ w.r.t. the graph norm of $d^{2}_{k}$, and so the infinite direct sum 
\be\label{hatgamma}
\op\hat\gamma_0 :\D(\op d_{k}^{2})\to \ell^{2}(\Z)\,.
\ee
is a well defined bounded operator.
Since $\gamma_{0}$ corresponds to $\op\hat\gamma_0 $ by partial Fourier transform, \eqref{hatgamma} does not define a trace map since it is not surjective: its range space is the strict subspace of $\ell^2(\Z)$ defined by 
$$h^{\frac12}(\Z):=\left\{\{s_{k}\}_{k\in\Z}\in\ell^{2}(\Z):\sum_{k\in\Z}|k|\,|s_{k}|^{2}<+\infty\right\}\simeq H^{\frac12}(\ci)\,.$$  This simple example shows that an infinite direct sum of trace maps can fail to be a trace map: the direct sum of the range spaces can be different from the range space of the sum.
\par
In Section 2 we provide a simple criterion which selects the right range space in order that 
the direct sum of trace maps is a trace map. Such a simple criterion uses a hypothesis involving the boundedness of operator-valued  sequences obtained composing the trace maps with their right inverses (see \eqref{sup}). Such a hypothesis seems a very strong one (indeed that allows an easy proof), however we show that always there exist right inverses such that 
\eqref{sup} holds true (see Lemma \ref{iota-tau}).  \par
In Section 3 we give an application to self-adjoint extensions of direct sums of symmetric operators and provide a couple of examples. We obtain that the methods here presented permit to obtain results equivalent to the ones recently obtained in \cite{MN12} and \cite{KM} using regularized boundary triplets (see Remark \ref{reg}). \par 
In Example 1 we determine the trace space for the evaluation map $f\mapsto\{f'(x_{n})\}_{n\in\NA}$  acting on functions $f\in H^{2}(\RE\backslash X)\cap H^{1}_{0}(\RE\backslash X)$ where $X=\{x_{n}\}_{n\in\NA}\subset\RE$, $x_{n}<x_{n+1}$. In this case Theorem \ref{lemma1} easily implies that the range space is a weighted $\ell^2$-space with weight $w_{n}=d^{-1}_{n}$, where $d_{n}:=x_{n+1}-x_{n}$. By Theorem \ref{estensioni} such a trace map can be used to define one-dimensional Schr\"odingier  operators with $\delta$ and $\delta'$ interaction supported on the discrete set $X$, thus providing a construction alternative to the one presented in \cite{KM}. \par 
In Example 2 we show that our criterion easily gives the correct trace space $H^{\frac12}(\ci)$ for the example provided at the beginning. Then we point out that the same criterion allows to prove that $H^{s}(\ci)$, $s=\frac12-\frac{\alpha}{1+\alpha}$, is (isomorphic to) the defect space of $\Delta^{\min}_{\alpha}$, $-1<\alpha<1$, the minimal realization of the Laplace-Beltrami operator $\Delta_{\alpha}:=\frac{\partial^{2}}{\partial x^{2}}-\frac{\alpha}{x}\,\frac{\partial\ }{\partial x}+x^{2\alpha}\,
\frac{\partial^{2}}{\partial \theta^{2}}$ corresponding to the degenerate/singular Riemannian metric $g_{\alpha}(x,\theta)
=\left(\begin{smallmatrix}1&0\\0&x^{-2\alpha}\end{smallmatrix}\right)$. We refer to the papers \cite{BL} and \cite{BP} for the almost-Riemannian geometric considerations leading to the study of  $\Delta_{\alpha}$ and to \cite{PP} for the classification of all self-adjoint extensions of  $\Delta^{\min}_{\alpha}$.
\end{section}

\begin{section}{Direct sums of abstract trace maps}\label{trace}
\noindent Let $\H_k$, $k\in\Z$, be a sequence of Hilbert spaces, with scalar product $\langle\cdot,\cdot\rangle_{k}$ and corresponding norm $\|\cdot\|_k$. 
On each $\H_k$ we consider a self-adjoint operator 
$$
A_k:\D(A_k)\subset \H_k\to \H_k
$$
and we denote by $\H_{(k)}$ the Hilbert space consisting of $\D(A_k)$ equipped with a scalar product $\langle \cdot,\cdot \rangle_{(k)}$ giving rise to a norm $\|\cdot\|_{(k)} $ equivalent to the graph one. 
\par
Let  $\h_k $, $k\in\Z$, be a sequence of  auxiliary Hilbert spaces with scalar product $[\cdot,\cdot]_{k} $ and corresponding norm 
$|\cdot |_{k} $. \par
Let
$$
\tau_k : \H_{(k)} \to \h_k \,,\quad k\in\Z\,,
$$ 
be a sequence of {\it abstract trace maps}, i.e. $\tau_{k}$ is a linear, continuous and surjective map 
such that its kernel $\K(\tau_{k} )$ is dense in $\H_{k}$. Since $\tau_{k}$ is continuous and surjective 
there exists a linear continuous right inverse
$$
\iota_{k}:\h_{k}\to \H_{(k)}\,,\quad \tau_{k}\iota_{k}=\uno
$$
(see e.g. \cite[Proposition 1, Section 6, Chapter 4]{Aub}). Since $\tau_{k}$ is surjective, $\iota_k$ 
is injective and so we can define a new scalar product on $\h_{k} $ by
$$[\phi_{k},\psi_{k}]_{(k)}:=[\iota_k^{*}\iota_k \phi_{k},\psi_{k}]_{k} \equiv\langle \iota_k \phi_{k},\iota_k \psi_{k}\rangle_{(k)}\,.
$$
It is immediate to check that $\h_{k}$ is complete w.r.t. the norm  $$|\phi_{k}|_{(k)}:=\|\iota_k\phi_{k}\|_{(k)}\equiv|(\iota_{k}^{*}\iota_k)^{1/2}\phi_{k}|_{k}\,.$$ Let us denote by $\h_{(k)}$ the Hilbert space given by $\h_{k}$ equipped with the scalar product $[\cdot,\cdot]_{(k)}$. We define
$$
\H:=\op\H_{k}\,,\quad\quad \H_{\circ}:=\op\H_{(k)}\,,
$$
$$ 
\h:=\op\h_{k}\,,\quad \quad
\h_{\circ}:=\op\h_{(k)}
$$
with corresponding norms $\|\cdot\|$, $\|\cdot\|_{\circ}$, $|\cdot|$, $|\cdot|_{\circ}$. \par 
We denote by $\n\cdot\n$ the operator norm of bounded linear operators.
\begin{theorem}\label{lemma1} 
Let $\iota_{k}$ be a linear continuous right inverse of $\tau_{k}$ and suppose that 
\begin{equation}\label{sup}
\sup_{k\in\Z}\,\n\iota_{k}\tau_{k}\n<+\infty\,.
\end{equation}
Then the linear map
$$
\tau:\H_{\circ}\to\h_{\circ}\,,\quad \tau(\op v_{k}):=
\op(\tau_{k}v_{k})
$$
is an abstract trace map, i.e. is continuous, surjective and its kernel $\K(\tau)$ is dense in $\H$.
\end{theorem}
\begin{proof} (continuity) Let $v=\op v_k \in\H_{\circ}$. Then
\begin{align*}
|\tau v|_{\circ}^{2}=&\sum_{k\in\Z}\|\iota_ {k}\tau_ {k}v_ {k}\|_{(k)}^{2} \le
\left(\sup_{k\in\Z}\,\n\iota_ {k}\tau_ {k}\n\right)^{2}\sum_{k\in\Z}\|v_ {k}\|_{ (k)}^{2} \\
=&\left( \sup_{k\in\Z}\,\n\iota_ {k}\tau_ {k}\n\right)^{2}\|v\|^{2}_{\circ}\,.
\end{align*}
(surjectivity) Given $\phi=\op\phi_ {k} \in\h_{\circ}$, let us define $v:=\op v_ {k} $ by 
$v_ {k}=\iota_k\phi_ {k}\in \H_{ (k)}$. Then $v\in \H_{\circ}$ by
$$\sum_{k\in\Z}\|v_ {k}\|^{2}_{ (k)} =\sum_{k\in\Z}\|\iota_k \phi_ {k}\|^{2}_{ (k)} =\sum_{k\in\Z}|\phi_ {k}|^{2}_{ (k)} =|\phi|^{2}_{\circ}\,.
$$
(density) Given $v:=\op v_ {k} \in \H$ and $\epsilon >0$, let $N_{\epsilon}\ge 0$ such that 
$\sum_{|k|> N_{\epsilon}}\|v_ {k}\|^{2}_ {k} \le \epsilon/2$. Since $\K(\tau_ {k})$ is dense in 
$\H_ {k}$, there exist $v_{k,\epsilon }\in\K(\tau_ {k})$ such that $\|v_ {k}-v_{k,\epsilon}\|^{2}_ {k}\le 2^{-|k|}(\epsilon/6) $. Define  
$v_{\epsilon}:=\underset{|k|\le N_{\epsilon}}{\oplus}v_{k,\epsilon} $. Then $v_{\epsilon}\in\K(\tau)$ and 
$$
\|v-v_{\epsilon}\|^{2}\le \sum_{|k|\le N_{\epsilon}}\|v_ {k}-v_{k,\epsilon}\|^{2}_ {k} +\frac{\epsilon}2\le \frac{\epsilon}6\sum_{k\in\Z}2^{-|k|}+\frac{\epsilon}2=\epsilon \,.
$$
\end{proof}
\begin{remark}
Notice that Theorem \ref{lemma1} holds true for any sequence of Hilbert spaces $\H_{(k)}$, $k\in\NA$, such that each $\H_{(k)}$ is densely embedded in $\H_{k}$. However our hypotheses 
$\H_{(k)}=\D(A_{k})$ permits to show that it is always possible to find right inverses $\iota_{k}$ such that hypothesis \eqref{sup} is satisfied (see Lemma \ref{iota-tau} below).
\end{remark}
 For any $z\in\rho( A_{k})$, let us define the following bounded linear operators:
$$
R_{k}(z):\H_{k}\to\H_{(k)}\,,\quad R_{k}(z):=(- A_{k}+z)^{-1}\,,
$$
$$
G_k(z):\h_{k}\to \H_{k}\,,\quad G_k(z):=(\tau_{k}R_k(\bar z))^{*}\,.
$$
By resolvent identity one has 
\begin{align}\label{res-id}
G_k(w)-G_k(z)=&(z-w)R_{k}(w)G_k(z)=(z-w)R_{k}(z)G_k(w)\,.
\end{align}
Now let us take $z=\pm i$ in the above definitions and pose
$$R_{k}:=(-A_{k}+i)^{-1}\,,\quad G_{k}:=G_{k}(-i)\,,\quad G_{k}^{+}:=G_{k}(i)\,,
$$
$$\Gamma_{k}(z):=
\tau_{k}\left(\frac{G_{k}+G_{k}^{+}}2-G_{k}(z)\right)\,.
$$
Then $z\mapsto \Gamma_{k}(z)$ is a Weyl function (equivalently a Krein's Q-function), i.e. it satisfied the identities  
$$
\Gamma_k(z)-\Gamma_k(w)=(z-w)G_k(\bar w)^{*}G_k(z)
$$
and
$$
\Gamma_k(z)^{*}=\Gamma_k(\bar z)\,.
$$
Therefore the set $$
Z_{k}:=\{z\in\rho(A_{k}):0\in\rho(\Gamma_k(z))\}\,.
$$
is not void: $\CO\backslash\RE\subseteq Z_{k}$ (see e.g. \cite[Theorem 2.1]{P08}).\par
Posing 
$$
\Gamma_{k}:=\Gamma_{k}(-i)\,,
$$
one has the identities
\be
G_{k}^{+}-G_{k}=2iR_{k}G_{k}\,,
\ee
\begin{equation}\label{GG} 
G_{k}^{*}G_{k}=-i\,\Gamma_{k}
\end{equation}
and so
\begin{equation}\label{inverse1}
\iota_k :\h_{k} \to \H_{(k)}\,,\quad \iota_{k}:=i R_{k}G_{k}\Gamma_{k}^{-1}=R_{k}G_{k}(G_{k}^{*}G_{k})^{-1}\,.
\end{equation}
is a linear bounded right inverse of $\tau_{k}$. Moreover, since 
$R_{k}:\H_{k}\to \H_{(k)}$ is unitary w.r.t. the scalar product 
$$
\langle u_{k},v_{k}\rangle_{(k)}:=\langle (-A+i)u_{k},(-A+i)v_{k}\rangle_{k}\,,$$  
one has
\begin{equation}\label{inverse2}
\iota_{k}^{*}\iota_{k}
=(G_{k}^{*}G_{k})^{-1}\,.
\end{equation}
\begin{lemma}\label{iota-tau} 
Let $\iota_{k}$ be defined as in \eqref{inverse1}. Then $\n\iota_{k}\tau_{k}\n= 1$.
\end{lemma}
\begin{proof}
By \eqref{inverse1} 
one has
$$
\|\iota_{k}\tau_{k}v_{k}\|_{(k)}
=\|G_{k}(G^{*}_{k}G_{k})^{-1}G^{*}_{k}(-A_{k}+i)v_{k}\|_{k}\,.
$$
Since the range of $G_{k}$ is closed one has the decomposition $\H_{k}=\R(G_{k})\oplus \K(G_{k}^{*})$ and so  $(-A_{k}+i)v_{k}=G_{k}\phi_{k}\oplus w_{k}$. Therefore $$
\|\iota_{k}\tau_{k}v_{k}\|_{(k)}=\|G_{k}\phi_{k}\|_{k}\le \|(-A_{k}+i)v_{k}\|_{k}=
 \|v_{k}\|_{(k)}\,.
$$  
If $v_{k}=R_{k}G_{k}\phi_{k}$ then $\|\iota_{k}\tau_{k}v_{k}\|_{(k)}=
 \|v_{k}\|_{(k)}$.
\end{proof}
\begin{remark}\label{variant}
In the case there exists $\lambda\in\cap_{k\in\Z}\,\rho(A_{k})\cap\RE$ the previous reasonings have the following variant. By \eqref{res-id}
there follows
\begin{align*}
&G_{k}(-i)^{*}G_{k}(-i)\\
=&G_{k}( \lambda)^{*}(\uno +(\lambda-i)R_{k}(i))(\uno +(\lambda+i)R_{k}(-i))G_{k}(\lambda)\,.
\end{align*}
and so
\begin{align*}
|G_{k}(-i)^{*}G_{k}(-i)\phi_{k}|_{k}
\le&\n\uno +(\lambda-i)R_{k}(i)\n^{2}
|G_{k}(\lambda)^{*}G_{k}(\lambda)\phi_{k}|_{k}\\
\le&\left(1+\sqrt{1+\lambda^{2}}\,\right)^{2}
|G_{k}(\lambda)^{*}G_{k}(\lambda)\phi_{k}|_{k}\,.
\end{align*}
Since $G_{k}(-i)^{*}G_{k}(-i)$ is injective by \eqref{GG}, this shows that $G_{k}(\lambda)^{*}G_{k}(\lambda)$ is injective. Since it is self-adjoint and its range is closed (since the range of $G_{k}(\lambda)$ is closed), $G_{k}(\lambda)^{*}G_{k}(\lambda)$ is a continuous bijection. Then $$
\iota_{k}:=R_{k}G_{k}(G_{k}^{*}G_{k})^{-1}\,,
$$ 
is a bounded right inverse of $\tau_{k}$, where in this case we used the notation 
$$
R_{k}:=(-A_{k}+\lambda)^{-1}\,, \quad G_{k}:=G_{k}(\lambda)\,.
$$
Moreover, by using the scalar product  
$$
\langle u_{k},v_{k}\rangle_{(k)}:=\langle (-A_{k}+\lambda)u_{k},(-A_{k}+\lambda)v_{k}\rangle_{k}\,,
$$
one gets
$$
\iota_{k}^{*}\iota_{k}=(G_{k}^{*}G_{k})^{-1}\,.
$$
and, proceeding as in the proof of lemma \ref{iota-tau},
$$
\n\iota_{k}\tau_{k}\n=1\,.
$$
\end{remark}
Theorem \ref{lemma1} has the following alternative version where one can still use the original trace space $\h$ as long as one regularizes the traces $\tau_{k}$:
\begin{theorem}\label{lemma3}
Let us define $r_k:=(G_{k}^{*}G_{k})^{1/2}$ and 
$$\ttau_{k} :\H_{(k)}\to\h_{k}\,,\quad \ttau_{k}:=r_k^{-1}\tau_{k} $$
Then the linear map
$$
\ttau:\H_{\circ}\to\h\,,\quad \ttau (\op v_{k}):=
\op(\ttau_{k}v_{k})
$$
is continuous, surjective and its kernel $\K(\ttau)=\K(\tau)$ is dense in $\H$.
\end{theorem}
\begin{proof} The proof is the same as in Theorem \ref{lemma1}.
It suffices to notice that $\tilde\iota_k:=\iota_kr_k$ is the right inverse of $\ttau_{k}$ 
and that $$(\tilde\iota_k)^{*}\tilde\iota_k=r_k\iota^{*}_{k}\iota_{k}r_k=r_k(G_{k}^{*}G_{k})^{-1}r_k=\uno\,.$$ 
\end{proof}
\begin{remark}\label{finite}
Notice that in this section $\Z$ can be replaced by any other countable set $N$ and that we can replace $[\cdot,\cdot]_{(k)}$ by a scalar product inducing an equivalent norm. Moreover, given a finite subset $F\subset N$, we can replace $\h_{(k)}$ by $\h_{k}$ for any $k\in F$.
\end{remark}
\end{section}
\begin{section}{Applications and Examples.} 
Let $S_{k}$, $k\in\Z$, be the sequence of symmetric operators defined by $S_{k}:=A_{k}|\K(\tau_{k})$,  where $A_{k}$ and $\tau_{k}$ are defined as in the previous section. Then $S:=\op S_{k}$ is a symmetric operator and $S=A|\K(\tau)$, where 
$A:=\op A_{k}$ and $\tau:=\op\tau_{k}$ is defined as in Theorem \ref{lemma1}. 
Here $\tau_{k}$ is considered as a map on $\H_{(k)}$ to $\h_{(k)}$, so that when calculating the adjoint $G_{(k)}(z)$ of $\tau_{k}(R_{k}(\bar z))$ one gets 
$$
G_{(k)}(z):=G_{k}(z)\iota_{k}^{*}\iota_{k}\,.
$$
Next Lemma shows that the direct sums $\op G_{(k)}(z)$ appearing in Theorem \ref{estensioni} below are well defined bounded operators:
\begin{lemma}\label{gz}
$$
\forall z\in\bigcap_{k\in\Z}\,\rho(A_{k})\,,\quad \sup_{k\in\Z}\,
\n G_{(k)}(z)\n<+\infty\,. 
$$
\end{lemma}
\begin{proof}
By \eqref{res-id} one has, posing $G_{(k)}:=G_{(k)}(-i)$,
$$
\n G_{(k)}(z)\n\le \n\uno-(i+z)R_{k}\n\,\n G_{(k)}\n\le (2+|z|)\,\n G_{(k)}\n\,.
$$
By \eqref{inverse2},
$$
\|G_{(k)}\phi_{k}\|_{k}=\|G_{k}\iota^{*}_{k}\iota_{k}\phi_{k}\|_{k}=\langle\iota^{*}_{k}\iota_{k}\phi_{k},G_{k}^{*}G_{k}\iota^{*}_{k}\iota_{k}\phi_{k}\rangle_{k}=|\phi_{k}|_{(k)}
$$
and so $$
\n G_{(k)}\n=1\,.
$$
\end{proof}
By Theorem \ref{lemma1} and by the results provided in \cite[Theorem 2.2]{P01} and \cite[Theorem 2.1]{P08} one gets the following
\begin{theorem}\label{estensioni} 
The set of self-adjoint extensions of $S$ is parametrized by couples 
$(\Pi,\Theta)$, where $\Pi$ is an orthogonal projection in $\h_{\circ}=\op \h_{ (k)} $ and $\Theta$ is a self-adjoint operator in the Hilbert space $\text{\rm Range}(\Pi)$. Denoting by $A_{\Pi,\Theta}$ the self-adjoint extension associated with $(\Pi,\Theta)$ one has
$$
A_{\Pi,\Theta}(\op v_ {k} )=\op \left(A_ {k}v^{\circ}_ {k}+\left(\text{\rm Re}(z_{\circ})G_ {(k)}^{\circ}+i\,
\text{\rm Im}(z_{\circ})G_ {(k)}^{\diamond}\right)\phi_ {k}\right) \,,
$$ 
\begin{align*}
\D(A_{\Pi,\Theta})=&\left\{\op v_{k}\in\H:v_{k}=v^{\circ}_{k}+G^{\circ}_{(k)}\phi_{k}\,,\ \op v^{\circ}_{k}\in\op\D(A_{k})\,,\right.\\ &\ \  \left.\op\phi_{k}\in\D(\Theta)\,,\ \Pi(\op\tau_{k}v^{\circ}_{k})=\Theta(\op\phi_{k})
\right\}\,.
\end{align*}
Moreover, for any $z\in(\cap_{k\in \Z}\,\rho(A_{k}))\cap\rho(A_{\Pi,\Theta})$, 
\begin{align*}
&(-A_{\Pi,\Theta}+z)^{-1}=\op(-A_k+z)^{-1}\\
&+\op G_{(k)}(z)\Pi\big(\Theta+\Pi\op \tau_{k}(G_{(k)}^{\circ}-G_{(k)}(z))\Pi\big)^{-1}
\Pi\op G^{*}_{(k)}(z)\,.
\end{align*}
Here $$G_{(k)}^{\circ}:=\frac12 (G_{(k)}(z_{\circ})+G_{(k)}(\bar z_{\circ}))\,,\quad
G_{(k)}^{\diamond}:=\frac12 (G_{(k)}(z_{\circ})-G_{(k)}(\bar z_{\circ}))$$
and $z_{\circ}\in\cap_{k\in\Z}\,\rho(A_{k})$.
\end{theorem}
\begin{remark} By the definition of $\D(A_{\Pi,\Theta})$ one has that $A_{\Pi,\Theta}$ is a direct sum if and only if both $\Pi$ and $\Theta$ are direct sums. \par
In the case $z_{\circ}\in\RE$ one has $G_{(k)}^{\diamond}=0$ and
$$
\tau_{k}(G^{\circ}_{(k)}-G_{(k)}(z))=z\,G_{k}(z_{\circ})^{*}G_{k}(z)\,\iota^{*}_{k}\iota_{k}
=z\,G_{k}(z)^{*}G_{k}(z_{\circ})\,\iota^{*}_{k}\iota_{k}
$$
In the case $0\in\cap_{k\in \Z}\,\rho(A_{k})$ one can take $z_{\circ}=0$, so that $$A_{\Pi,\Theta}(\op v_{k})=\op A_{k}v_{k}^{\circ}\,.$$ 
\end{remark}
\begin{remark}\label{reg}
Theorem \ref{estensioni} has an alternative version in the case one uses the trace map furnished by Lemma \ref{lemma3}. In this case the extension parameter $(\Pi,\Theta)$ is such that  $\Pi$ is an orthogonal projection in $\h=\op \h_{k}$ and $\Theta$ is a self-adjoint operator in the Hilbert space associated with $\Pi$. The statement of Theorem \ref{estensioni} remains unchanged replacing $\tau_{k}$ with $\tilde \tau_{k}$ and $G_{(k)}(z)$ with $$\tilde G_{k}(z):=G_{(k)}(z)r_k=G_{k}(z) r_k^{-1}=(\tilde\tau_{k}R_{k}(\bar z))^{*}\,.
$$ 
Also in this case the norm of $\tilde G_{(k)}(z):\H_{k}\to\h_{k}$ is bounded uniformly in $k\in\Z$ for any $z\in\cap_{k\in\Z}\,\rho(A_{k})$. 
\end{remark}
\begin{remark}\label{reg} 
By \cite{P04} and \cite[Section 4]{P08}, Theorem \ref{estensioni} and Lemma \ref{lemma3} provide results equivalent to the ones that can be obtained using Boundary Triplet Theory.  Let us for simplicity take $z_{\circ}=i$. Then (see \cite[Theorem 3.1]{P04})
$$
\D(S_{k}^{*})=\{v_{k}\in\H:v_{k}=v^{\circ}_{k}+G^{\circ}_{k}\phi_{k}\,,\ v^{\circ}_{k}\in\D(A_{k})\,,\ \phi_{k}\in\h_{k}\}\,,
$$
$$
S^{*}_{k}:\D(S_{k}^{*})\subseteq\H_{k}\to\H_{k}\,,\quad S^{*}_{k}v_{k}:=A_{k}v_{k}+R_{k}G_{k}\phi_{k}\,,
$$
and the triple $\{\h_{k},\beta_{k,0},\beta_{k,1}\}$, where
$$
\beta_{k,0}:\D(S^{*}_{k})\to\h_{k}\,,\quad \beta_{k,0}v_{k}:=\tau_{k}
v^{\circ}_{k}\,,
$$
$$
\beta_{k,1}:\D(S^{*}_{k})\to\h_{k}\,,\quad \beta_{k,1}v_{k}:=\phi_{k}\,,
$$
is a boundary triple for $S^{*}_{k}$, i.e. $\beta_{k,1}$ and $\beta_{k,2}$ are surjective and  
satisfy the Green-type identity
\be\label{green}
\langle S_{k}^{*}u_{k}, v_{k}\rangle_{k}-\langle u_{k}, S_{k}^{*}v_{k}\rangle_{k}=
[\beta_{1,k}u_{k},\beta_{k,0}v_{k}]_{k}-[\beta_{k,0}u_{k},\beta_{k,1}v_{k}]_{k}\,.
\ee
Moreover the Weyl function of $S_{k}$ is 
$$
M_{k}(z)=\tau_{k}\left(\frac{G_{k}+G^{+}_{k}}2-G_{k}(z)\right)
$$ (see \cite[Theorem 3.1]{P04}).
By \eqref{green} there follows that  $\{\h_{k},r_{k}\beta_{k,1},r_{k}^{-1}\beta_{k,2}\}$, where $r_{k}$ is defined in Lemma \ref{lemma3}, is a boundary triple for $S^{*}_{k}$ as well with 
Weyl function $r_{k}^{-1}M_{k}(z)\,r_{k}^{-1}$. 
\par
By Lemma \ref{lemma3} and \cite[Theorem 1.6]{P04} one gets
\begin{align*}
\D(S^{*})
=&\left\{\op v_{k}:v_{k}=v^{\circ}_{k}+\tilde G^{\circ}_{k}\phi_{k}\,,\ \op v^{\circ}_{k}\in\op\D(A_{k})\,,\ \op\phi_{k}\in\h\right\}\\
\equiv&\left\{\op v_{k}:v_{k}=v^{\circ}_{k}+G^{\circ}_{k}\psi_{k}\,,\ \op v^{\circ}_{k}\in\op\D(A_{k})\,,\ \op r_{k}\psi_{k}\in\h\right\}
\end{align*}
and the triple $\{\h,\tilde\beta_{0},\tilde\beta_{1}\}
$ is a boundary triple for $S^{*}=\op S_{k}^{*}$ with Weyl function $\op (r_{k}^{-1}M_{k}(z)\,r_{k}^{-1})$, where
$$
\tilde \beta_{0}
:\D(S^{*})\to\h\,,\quad 
\tilde \beta_{0}(\op v_{k}):=\op (r_{k}^{-1}\beta_{0,k})(\op v_{k})\equiv
\op (r_{k}^{-1}\tau_{k}v^{\circ}_{k})\,,
$$
$$
\tilde \beta_{1}
:\D(S^{*})\to\h\,,\quad 
\tilde \beta_{1}(\op v_{k}):=\op (r_{k}\beta_{1,k})(\op v_{k})\equiv\op (r_{k}\psi_{k})
\equiv \op\phi_{k}.
$$
Let us notice that the norm of $r_{k}^{-1}\tau_{k}\equiv\tilde \tau_{k}:\H_{(k)}\to\h_{k}$ is bounded uniformly in $k\in\Z$ by Lemma \ref{lemma3}, and so $\tilde \beta_{0}$ is well defined. $\tilde \beta_{1}$ is well defined as well by the definition of $\D(S^{*})$.\par
In conclusion this provides results on direct sums of regularized boundary triplets of the kind  recently obtained in \cite[Section 5]{MN12}, \cite[Section 3]{KM}, \cite[Section 2]{CMP}. 
\end{remark}

\vskip8pt\noindent
{\bf Example 1.} Given $\{x_{n}\}_{n\in\NA}\subset\RE$, $x_{n}<x_{n+1}$, let $\H_{n}:=L^{2}(I_{n})$, $n\ge 0$, where $I_{0}=(-\infty,x_{1}]$ and $I_{n}=[x_{n},x_{n+1}]$, $n\in\NA$. Let 
$$
A_{n}:\D(A_{n})\subset  L^{2}(I_{n})\to 
L^{2}(I_{n})\,,\quad A_{n}u=u''\,,\quad n\ge 0\,,
$$
$$
\D(A_{0}):=\{u\in L^{2}(I_{0})\cap C^{1}(I_{0}): u''\in L^{2}(I_{0})\,,\ u(x_{1})=0\}\,,
$$
$$
\D(A_{n}):=\{u\in C^{1}(I_{n}): u''\in L^{2}(I_{n})\,,\ u(x_{n})=u(x_{n+1})=0\}\,,\quad n\in\NA\,,
$$
$$
\tau_{0}:\H_{(0)}\to\CO\,,\quad \tau_{0}u:=-u'(x_{1})\,.
$$
$$
\tau_{n}:\H_{(n)}\to\CO^{2}\,,\quad \tau_{n}u:=(u'(x_{n}),-u'(x_{n+1}))\quad n\in\NA\,.
$$
 For any $n\ge 0$, the map $\tau_{n}$ is continuous, surjective and has a kernel dense in $L^{2}(I_{n})$.\par
By Remark \ref{finite} we can suppose $n\not=0$ and, since $0\in\cap_{n>0}\,\rho(A_{n})$, we can use the results provided in Remark \ref{variant} with $\lambda=0$. \par 
The kernel of $(-A_{n})^{-1}$, $n>0$, is given by 
\begin{align*}
&K_{n}(x,y)\\=&\frac1{d_{n}}\left((x_{n+1}-x)(y-x_{n})\theta(x-y)+(x-x_{n})(x_{n+1}-y)\theta(y-x)\right)\,,
\end{align*}
where $\theta$ denotes Heaviside's function and $d_{n}:=x_{n+1}-x_{n}$. Therefore
$$
(G_{n}\xi)(x)=\frac{1}{d_{n}}\,(\xi_{1}(x_{n+1}-x)+\xi_{2}(x-x_{n}))\,,\quad \xi\equiv(\xi_{1},\xi_{2})\,,
$$
$$
G^{*}_{n}u\equiv\frac1{d_{n}}\,\left(\int_{x_{n}}^{x_{n+1}}(x_{n+1}-x)u(x)\,dx,
\int_{x_{n}}^{x_{n+1}}(x-x_{n+1})u(x)\,dx\right)
$$
and so by straightforward calculations one gets that $G_{n}^{*}G_{n}:\CO^{2}\to\CO^{2}$ corresponds to the positive-definite  matrix
$$G_{n}^{*}G_{n}\equiv {d_n}\,\left[\begin{matrix}1/3&1/6\\1/6&1/3\end{matrix}\right]\,.
$$ 
In conclusion on $\h_{(n)}=\CO^{2}$ we can put  the equivalent scalar product
$$
[\xi,\zeta]_{(n)}:=\frac {\xi\cdot\zeta}{{d_n }}\,.
$$
Hence, by Theorem \ref{lemma1}, denoting by 
$\ell^{2}_{d}(\NA)$ the weighted 
$\ell^{2}$-space 
$$
\ell^{2}_{d}(\NA):=\left\{\{s_{n}\}_{n\in\NA}:\sum_{n\in \NA}\frac {|s_{n}|^{2}}{d_{n}}<+\infty\right\}\,,
$$
one gets that 
\be\label{te}
\tau:=\tau_{0}\oplus(\opn\tau_{n}):\H_{0}\oplus(\opn\H_{(n)})\to\CO\oplus
\ell^{2}_{d}(\NA)\oplus\ell^{2}_{d}(\NA)\ee
is continuous, surjective and has a kernel dense in $L^{2}(-\infty,x_{\infty})$, $x_{\infty}:=\sup_{n\in\NA}x_{n}$.
Notice that $\ell^{2}_{d}(\NA)=\ell^2(\NA)$ if and only if 
$$
0<d_{*}:=\inf_{n\in\NA}d_{n}\le d^{*}:=\sup_{n\in\NA}d_{n}<+\infty\,.
$$
By using Theorem \ref{estensioni} with trace map $\tau$ defined in \eqref{te}, one gets the same kind of self-adjoint extensions given in \cite{KM} (the case in which $0<d_{*}\le d^{*}<+\infty$ has been studied in \cite{K79}). Such extensions describe one-dimensional Schr\"odinger operators in $L^{2}(-\infty,x_{\infty})$ with $\delta$ and $\delta'$ interactions supported on the discrete set $X=\{x_{n}\}_{n\in\NA}$. These operators have been studied in \cite[Chapters III.2 and III.3]{AGH-KH}, when $0<d_{*}\le d^{*}<+\infty$ and $x_{\infty}=+\infty$, and in \cite{KM} when $d^{*}<+\infty$. Analogous considerations, with $A_{n}$ given by the one-dimensional Dirac operator with Dirichlet boundary conditions on the interval $I_{n}$, lead to self-adjoint extension describing one-dimensional Dirac operators with    $\delta$ and $\delta'$ interactions on the discrete set $X=\{x_{n}\}_{n\in\NA}$ (see \cite[Appendix J]{AGH-KH}, for the case $X$ in which is a finite set and \cite{CMP} for the general case).
\vskip8pt
\par\noindent
{\bf Example 2.} At first let us check that applying Thereom \ref{lemma1} to the example given in the introduction one gets the right trace space $\h_{\circ}=h^{\frac12}(\Z)$. Hence here $A_{k}=d^{2}_{0}-k^{2}$. By Remark \ref{finite} we can suppose $k\not=0$ and, since  $0\in\cap_{k\in\Z\backslash\{0\}}\,\rho(A_{k})$, we can use the results provided in Remark \ref{variant} with $\lambda=0$. Since the kernel of $(-d^{2}_{0}+z^{2})^{-1}$, Re$(z)>0$, is given  by 
$$
K(z;x,y)=\frac{e^{-z\,|x-y|}-e^{-z\,(x+y)}}{2z}\,,
$$
one easily gets
$$
G_{k}^{*}\equiv \hat\gamma_{0} (-d^{2}_{0}+k^{2})^{-1}:L^{2}(\RE_{+})\to\CO \,,\quad G^{*}_{k}f=\int_{0}^{\infty} e^{-|k|\,x} f(x)\,dx
$$
and so $G_{k}^{*}G_{k}:\CO\to\CO$ is given by the multiplication by the real number
$$
G_{k}^{*}G_{k}\equiv\int_{0}^{\infty}e^{-2|k|\,x}\,dx=
\frac{1}{2\,|k|}\,.
$$
Therefore $\h_{(k)}=\CO$ is equipped with the scalar product 
$$
[\xi,\zeta]_{(k)}:=|k|\,\xi\cdot\zeta 
$$ 
and so
$$
\h_{\circ}=\CO\oplus\left(\underset{k\in\Z\backslash\{0\}}{\oplus} \h_{(k)}\right)=\left\{\{s_{k}\}_{k\in\Z}\in\ell^{2}(\Z):\sum_{k\in\Z}|k|\,|s_{k}|^{2}<+\infty\right\}
\,.
$$
By using Theorem \ref{estensioni} with trace map 
$$
\gamma_{0}=\op\hat\gamma_{0} :\op\D_{0}\to h^{\frac12}(\Z)\,,
$$
then one can determine all self-adjoint extensions of the minimal Laplacian on $\M_0 $.\par 
Such an example can be generalized in the following way: let $\M_\alpha  $ be $\RE_{+}\times\ci$ endowed with the singular/degenerate Riemannian metric 
$$
g_{\alpha}(x,\theta)=\left(\begin{matrix} 1&0\\0&x^{-2\alpha}\end{matrix}\right)\,,\quad\alpha\in\RE\,.
$$ 
The Riemannian volume form corresponding to $g_{\alpha}$ is $d\omega=x^{-\alpha}dxd\theta$ and so we denote by  $L^{2}(\M_\alpha  )$ be the Hilbert space  
$$
L^{2}(\M_\alpha  ):=\{u:\RE_{+}\times\ci\to\CO:\smallint_{0}^{2\pi}\smallint_{0}^{+\infty}|u(x,\theta)|^{2}\,x^{-\alpha}dxd\theta <+\infty\}\,.
$$
In \cite{BP} it is shown that the minimal realization 
$$\Delta_{\alpha}^{\min}:C^{\infty}_{c}(\M_\alpha  )\subset L^{2}(\M_\alpha  )\to 
L^{2}(\M_\alpha  ) 
$$
of the Laplace-Beltrami operator  
\be\label{LB}
\Delta_{\alpha}:=\frac{\partial^{2}}{\partial x^{2}}-\frac{\alpha}{x}\,\frac{\partial\ }{\partial x}+x^{2\alpha}\,
\frac{\partial^{2}}{\partial \theta^{2}}
\ee
corresponding to $g_{\alpha}$
is essentially self-adjoint whenever $\alpha\notin(-3,1)$, has deficiency indices $(1,1)$ whenever $\alpha\in(-3,-1]$ and has infinite deficiency indices whenever $\alpha\in(-1,1)$. Therefore, in order to determine and then study all self-adjoint realizations of $\Delta^{\min}_{\alpha}$, $-1<\alpha<1$, by Theorem \ref{estensioni} one needs to characterize the range space of the trace map 
$$
\gamma _{\alpha}u(\theta):=\lim_{x\downarrow 0}\, x^{-\alpha}\,\frac{\partial u}{\partial x}(x,\theta)
$$
acting on function in the domain of the Friedrichs extensions $\Delta^{D}_{\alpha}$ 
(corresponding to Dirichlet boundary conditions at  $\ci$) of $\Delta^{\min}_{\alpha}$ (see \cite{PP}).  Let us sketch here a proof in the case $0<\alpha<1$, referring to \cite{PP} for more details and for the (more involved but still using Theorem \ref{lemma1}) proof that holds in the case $-1<\alpha<1$. \par
By partial Fourier transform one gets 
$$
L^{2}(\M_\alpha  )=\op L^{2}_{w}(\RE_{+})\,,\quad \Delta^{D}_{\alpha}=\op (d^{2}_{\alpha}-k^{2}q_{\alpha})\,,
$$
where $L^{2}_{w}(\RE_{+})$ is the weighted $L^{2}$ space
$$
L^{2}_{w}(\RE_{+}):=\{f:\RE_{+}\to\CO:\int_{0}^{\infty}|f(x)|^{2} x^{-\alpha}dx<+\infty\}\,,
$$
and
$$
(d^{2}_{\alpha}-k^{2}q_{\alpha}):\D_{\alpha,k}\subset L_{w}^{2}(\RE_{+})\to L_{w}^{2}(\RE_{+})
\,,$$
$$
d^{2}_{\alpha}f(x):=f''(x)-\frac{\alpha}{x}\,f'(x)\,,\quad q_{\alpha}(x)=x^{2\alpha}\,,
$$
$$
\D_{\alpha,k}:=\{f\in L_{w}^{2}(\RE_{+})\cap C^{1}(\overline\RE_{+}):  (d^{2}_{\alpha}-k^{2}q_{\alpha})\in L_{w}^{2}(\RE_{+})\,,\ f(0)=0\}\,.
$$
By Remark \ref{finite} we can suppose $k\not=0$ and, since  $0\in\cap_{k\in\Z\backslash\{0\}}\,\rho(A_{k})$, $A_{k}=d^{2}_{\alpha}-k^{2}q_{\alpha}$, whenever $0<\alpha<1$, we can use the results provided in Remark \ref{variant} with $\lambda=0$. Since $f_{\xi}\equiv G_{k}\xi$ solves the boundary value problem
$$
\begin{cases} 
f_{\xi}''(x)-\frac{\alpha}{x}\,f_{\xi}'(x)-k^{2}x^{2\alpha}f_{\xi}=0\\
f_{\xi}(0)=\xi\,,
\end{cases}
$$
one gets $$(G_{k}\xi)(x)=\xi \exp\left(- \frac{|k|x^{\alpha+1}}{\alpha+1}\right)\,.
$$
Therefore $G_{k}^{*}G_{k}:\CO\to\CO$ is given by the multliplication by the real number
$$
G_{k}^{*}G_{k}\equiv\int_{0}^{\infty} e^{- 2\,\frac{|k|x^{\alpha+1}}{\alpha+1}} x^{-\alpha}
dx=|k|^{\frac{\alpha-1}{\alpha+1}}\int_{0}^{\infty} e^{- 2\,\frac{x^{\alpha+1}}{\alpha+1}} x^{-\alpha}
dx
$$
and so $\h_{(k)}=\CO$ is equipped with the scalar product
$$
[\xi,\zeta]_{(k)}:=|k|^{\frac{1-\alpha}{1+\alpha}}\,\xi\cdot\zeta\,.
$$
Thus by Theorem  \ref{lemma1}  the range space of $\gamma _{\alpha}$ (i.e. the defect space of $\Delta^{\min}_{\alpha}$) is given by the fractional Hilbert-Sobolev space 
$$
H^s(\ci)\simeq h^s(\Z):=
\left\{\{s_{k}\}_{k\in\Z}\in\ell^{2}(\Z):
\sum_{k\in\Z}|k|^{2s}\,|s_{k}|^{2}<+\infty\right\}\,,
$$
where
$s=\frac12-\frac{\alpha}{1+\alpha}$.

\end{section}

\vskip20pt\p
{\bf Acknoledgements.} I thank Ugo Boscain and Dario Prandi for the stimulating discussions which inspired this work.


\begin{thebibliography}{99}

\bibitem{AGH-KH} S. Albeverio, F. Gesztesy, R. Hoegh-Krohn, H. Holden: {\it Solvable Models in Quantum Mechanics}, second ed., Amer. Math. Soc. Chelsea, 2005, with appendix written by P. Exner.
\bibitem{Aub} J.-P. Aubin: {\it Applied Functional Analysis}, John Wiley, 1979.
\bibitem{BL} U. Boscain, C. Laurent: The Laplace-Beltrami operator in almost-Riemannian Geometry. {\it Annales de l'Institut Fourier}, {\bf 63}  (2013), 1739-1770.
\bibitem{BP} U. Boscain, D. Prandi: The Laplace-Beltrami Operator on Conic and Anticonic-type Surfaces. 	{\tt arXiv:1305.5271}.
\bibitem{CMP} R. Carlone, M.M. Malamud, A. Posilicano: On the spectral theory of Gesztesy-{\v S}eba realizations of 1-D Dirac operators with point interactions on a discrete set. 
\emph{J. Differential Equations}, \textbf{254} (2013), 3835-3902.

\bibitem{K79} A.N. Kochubei: Symmetric operators and nonclassical spectral problems, \emph{Math. Notes}, \textbf{25} (1979), 425--434.


\bibitem{KM} A.S. Kostenko, M.M. Malamud: 1-D
Schr\"{o}dinger operators with local point interactions on a discrete set,
 \emph{J. Differential Equations}, \textbf{249} (2010), 253 - 304.
 
 

\bibitem{MN12} M.M. Malamud, H. Neidhardt: Sturm-Liouville boundary value
problems with operator potentials and  unitary equivalence, \emph{J. Differential Equations},  \textbf{252} (2012), 5875-5922.


\bibitem{P01}  A. Posilicano: Self-adjoint extensions by additive perturbations.  
{\it Ann. Sc. Norm. Super. Pisa Cl. Sci.} (5), {\bf 2} (2003), 1-20.

\bibitem{P04} A. Posilicano: Boundary Triples and Weyl Functions for Singular Perturbations of Self-Adjoint Operators, \emph{Methods Funct. Anal. Topology}, \textbf{10} (2004), 57--63.

\bibitem{P08} A. Posilicano: Self-Adjoint Extensions of
  Restrictions. \emph{Operators and Matrices} 
\textbf{2} (2008), 483-506.

\bibitem{PP} A. Posilicano, D. Prandi. Self-adjoint realizations of the Laplace-Beltrami operator on conic and anti-conic surfaces. In preparation.
\end{thebibliography}
\end{document}